\documentclass[english,aps,manuscript]{revtex4}
\usepackage[T1]{fontenc}
\usepackage[latin9]{inputenc}
\usepackage{amsthm}
\usepackage{amsmath}
\usepackage{amssymb}

\makeatletter
\@ifundefined{textcolor}{}
{%
 \definecolor{BLACK}{gray}{0}
 \definecolor{WHITE}{gray}{1}
 \definecolor{RED}{rgb}{1,0,0}
 \definecolor{GREEN}{rgb}{0,1,0}
 \definecolor{BLUE}{rgb}{0,0,1}
 \definecolor{CYAN}{cmyk}{1,0,0,0}
 \definecolor{MAGENTA}{cmyk}{0,1,0,0}
 \definecolor{YELLOW}{cmyk}{0,0,1,0}
 }
\theoremstyle{plain}
\theoremstyle{plain}
\newtheorem{thm}{Theorem}
  \theoremstyle{remark}
  \newtheorem*{acknowledgement*}{Acknowledgement}

\@ifundefined{definecolor}
 {\usepackage{color}}{}
\usepackage{amsthm,amsfonts,url}
\makeatother

\makeatother

\makeatother

\makeatother

\makeatother

\makeatother

\makeatother

\usepackage{babel}

\makeatother

\usepackage{babel}

\makeatother

\usepackage{babel}

\makeatother

\usepackage{babel}

\begin{document}

\title{Unambiguous discrimination of linearly independent pure quantum states:
Optimal average probability of success }

\author{Somshubhro Bandyopadhyay}

\email{som@jcbose.ac.in}

\affiliation{Department of Physics and Center for Astroparticle Physics and Space
Science, Bose Institute, Block EN, Sector V, Bidhan Nagar, Kolkata
700091}
\begin{abstract}
We consider the problem of unambiguous (error-free) discrimination
of $N$ linearly independent pure quantum states with prior probabilities,
where the goal is to find a measurement that maximizes the average
probability of success. We derive an upper bound on the optimal average
probability of success using a result on optimal local conversion
between two bipartite pure states. We prove that for any $N\geq2$
an optimal measurement in general saturates our bound. In the exceptional
cases we show that the bound is tight, but not always optimal. 
\end{abstract}
\maketitle
One of the consequences of the superposition principle is that quantum
states could be nonorthogonal, which restricts our ability to reliably
determine the state of a quantum system even when the set of possible
states is known. Thus a fundamental problem in quantum mechanics is
to determine how well quantum states can be distinguished from one
another (see \cite{Chefles-2000,Barnett-Croke} for reviews). In its
simplest form, the problem is defined as follows: A quantum system
is prepared in one of $N$ known pure states $\left|\psi_{1}\right\rangle ,\left|\psi_{2}\right\rangle ,\dots,\left|\psi_{N}\right\rangle $
with associated probabilities $p_{1},p_{2},\dots,p_{N}$, where $0<p_{i}<1$
for every $i$ and $\sum_{i=1}^{N}p_{i}=1$. We do not know which
state the system is in, but wish to identify it. If the states are
mutually orthogonal, the solution is straightforward. However, if
the states are not mutually orthogonal, then quantum mechanics forbids
us from distinguishing them perfectly. Therefore, the objective is
to devise a measurement strategy that is optimal according to some
reasonable quantifier of distinguishability. This scenario is typical
in quantum information theory, especially in quantum communications
and quantum cryptography. 

In this paper we consider how well a given set of pure states can
be discriminated without error. This measurement strategy, known as
unambiguous discrimination, seeks \emph{certain} knowledge of the
state of the system balanced against a probability of failure. Since
no error is permitted, in addition to the measurement outcomes that
correctly identify the input state, an inconclusive outcome, which
is not informative, must be allowed. That is, either the input state
is correctly detected or the outcome is inconclusive, in which case
we do not learn anything about the state. It may be noted that in
other strategies such as minimum error discrimination \cite{Chefles-2000,Barnett-Croke}
and maximum confidence measurements \cite{Barnett-Croke}, we cannot
in general be completely sure of the identity of the input state. 

Unambiguous discrimination of pure quantum states is possible if and
only if the states are linearly independent \cite{Chefles-1998}.
This assumption will therefore hold throughout this paper. The measurement
is described by a POVM $\boldsymbol{\Pi}=\left\{ \Pi_{k}\right\} $
with $N+1$ outcomes. where $\Pi_{k}\geq0$ and $\sum_{k=1}^{N+1}\Pi_{k}=I$.
The POVM elements $\left\{ \Pi_{k}\left|k=1,\dots,N\right.\right\} $
are associated with success and satisfy, \begin{eqnarray}
\left\langle \psi_{i}\left|\Pi_{j}\right|\psi_{j}\right\rangle  & = & \gamma_{i}\delta_{ij},\;\;\forall i,j=1,\dots,N\label{unambiguous-M}\end{eqnarray}
where $\gamma_{i}$ is the probability of successfully detecting the
state $\left|\psi_{i}\right\rangle $. Note that Eq.$\,$(\ref{unambiguous-M})
implies that if the system is in state $\left|\psi_{i}\right\rangle $,
the outcome $j\neq i$ for $j=1,\dots,N$ will never occur. The operator
$\Pi_{N+1}=I-\sum_{i=1}^{N}\Pi_{i}$ corresponds to an inconclusive
outcome. Notice that the set of individual success probabilities $\left\{ \gamma_{1},\dots,\gamma_{N}\right\} $
is determined only by our choice of POVM. Thus for a given measurement
$\boldsymbol{\Pi}$, the average probability of success is defined
as\begin{eqnarray}
P\left(\boldsymbol{\Pi}\right) & = & \sum_{i=1}^{N}p_{i}\gamma_{i}.\label{avg-prob-success}\end{eqnarray}
The goal is to find a measurement that maximizes the average probability
of success. In particular, we are interested in the following quantity:\begin{eqnarray}
P_{\mbox{opt}} & = & \max_{\left\{ \boldsymbol{\Pi}\right\} }P\left(\boldsymbol{\Pi}\right)=\sum_{i=1}^{N}p_{i}\gamma_{i}^{\mbox{opt}}\label{opt-prob-success}\end{eqnarray}
where the optimal solution $\boldsymbol{\gamma}^{\mbox{opt}}=\left\{ \gamma_{i}^{\mbox{opt}}\left|i=1,\dots,N\right.\right\} $
is the set of individual success probabilities maximizing the average
probability of success. The optimal solution is known only for $N=2$
\cite{Ivanovic-1987,Peres-1988,Dieks-1988,Jaeger-Shimony-1995}, and
special cases for $N\geq3$ \cite{Bergou-Fut-Feld-2012,Pang-Wu-2009,Sugimoto-et-al-2010,Sun-et-al-2001,Peres-terno-1998}.
General results include lower \cite{Duan-Guo-98,Sun-2002} and upper
\cite{Zhang-2001} bounds on the average probability of success, a
solution for $N$ equi-probable symmetric states \cite{Chefles-Barnett-1995},
a formulation of the problem as a semi-definite program with results
for symmetric and geometrically uniform states \cite{Eldar-2003},
characterization of optimal solutions \cite{Pang-Wu-2009}, a graphic
method for finding and classifying optimal solutions \cite{Bergou-Fut-Feld-2012},
and solution for equidistant states \cite{Roa-etal-2011}. 

Before we state our results it is necessary to briefly review all
possible classes of optimal solution \cite{Pang-Wu-2009}, precise
definitions of which are given in the appendix. For a given set of
$N$ linearly independent pure states let $\mathcal{R}$ be the set
of all candidate optimal solutions. This set, said to be the critical
feasible region, is an $\left(N-1\right)$-dimensional region (hypersurface)
in the $N$-dimensional real vector space $\mathbb{R}^{N}$, and is
completely determined by the input states and the constraints imposed
by the problem. Once we specify the prior probabilities, the optimal
solution, which is an element of $\mathcal{R}$ becomes unique in
the sense that there is no other solution that is also optimal for
the same set of prior probabilities. Different sets of prior probabilities
in general lead to different optimal solutions within the set $\mathcal{R}$. 

The optimal solution is either an interior or a boundary point of
$\mathcal{R}$. If it is an interior point then it means that the
optimal measurement is able to discriminate all states, i.e., for
every $i$, $0<\gamma_{i}^{\mbox{opt}}\leq1$. On the other hand,
if it is a boundary point, then at least one of the optimal individual
success probabilities is zero. We say that an interior point is nonsingular
if the solution is nondegenerate, i.e., it can be the optimal solution
only for an unique set of prior probabilities. An interior point can
also be singular if the solution is degenerate, i.e., it can be the
optimal solution for different sets of prior probabilities. It should
be noted that interior singular points are exceptions and may not
even exist for a given set of states. Thus there are only three possible
classes of optimal solution: interior nonsingular, interior singular
and boundary. 

Using the conditions in \cite{Pang-Wu-2009}, it is easy to show that
for a given set of states, every interior nonsingular point is the
optimal solution for some set of prior probabilities. Noting that
the critical feasible region is of dimension $N-1$, the dimension
of the interior part is also $N-1$, whereas the dimensions of the
boundary regions are strictly less than $N-1$. Therefore, for almost
all assignments of prior probabilities, the optimal solution will
be an interior nonsingular point. In other words, for any given instance
of an unambiguous state discrimination problem, the optimal solution
in general will be an interior nonsingular point of the critical feasible
region. 

In this work we derive an upper bound on the optimal average probability
of success using a result \cite{Lo-Popescu-2001,Vidal-1999} on optimal
local conversion between two bipartite pure states. We prove that
the bound is saturated when the optimal solution is an interior nonsingular
point of the critical feasible region, which is the set of all candidate
optimal solutions. From the previous argument we therefore conclude
that for any given set of $N\geq2$ linearly independent pure states
with prior probabilities, the upper bound in general equals the optimal
average probability of success. 

When the optimal solution is either an interior singular point or
a boundary point, we show that the upper bound is tight. However,
we also show that it is not achieved in general by an optimal boundary
solution. The question, whether an optimal solution that is an interior
singular point always saturates our bound remains open.

We begin by obtaining an upper bound on the optimal average probability
of success. 
\begin{thm}
Suppose a quantum system is prepared in one of the linearly independent
pure states $\left|\psi_{1}\right\rangle ,\dots,\left|\psi_{N}\right\rangle $
with prior probabilities $p_{1},\dots,p_{N}$ respectively, where
$0<p_{i}<1$ for every $i$ and $\sum_{i=1}^{N}p_{i}=1$. For an optimal
unambiguous state discrimination measurement, the average success
probability $P_{\mbox{opt}}$ is bounded by
\end{thm}
\begin{eqnarray}
P_{\mbox{opt}} & \leq & \min_{\left\{ \theta_{j}\right\} }\left\Vert \sum_{j=1}^{N}\sqrt{p_{j}}e^{i\theta_{j}}\left|\psi_{j}\right\rangle \right\Vert ^{2}.\label{thm-1}\end{eqnarray}
Essentially we are required to minimize the norm of the vector $\sum_{j=1}^{N}\sqrt{p_{j}}e^{i\theta_{j}}\left|\psi_{j}\right\rangle $
with respect to the real parameters $\left\{ \theta_{j}\left|j=1,\dots,N\right.\right\} $
we are free to vary. Because of this we can always set one of the
$\theta_{i}$, say, $\theta_{1}$, equal to zero and minimize the
norm with respect to the remaining $N-1$ parameters. However, it
is often useful to express inequality (\ref{thm-1}) in a form where
the parameters defining the inner products of the states become explicit.
Let $\left\langle \psi_{i}\left|\psi_{j}\right.\right\rangle =\left|\left\langle \psi_{i}\left|\psi_{j}\right.\right\rangle \right|e^{\phi_{ij}},\; i<j$,
We then have,\begin{equation}
P_{\mbox{opt}}\leq1+\min_{\left\{ \theta_{i}\right\} }\sum_{1\leq i<j\leq N}2\sqrt{p_{i}p_{j}}\left|\left\langle \psi_{i}\left|\psi_{j}\right.\right\rangle \right|\cos\left(\theta_{j}-\theta_{i}+\phi_{ij}\right).\label{eq:P-opt-2}\end{equation}
We shall use (\ref{eq:P-opt-2}) in the examples given later in the
paper and appendix. 

The proof of the theorem relies on two facts. First, any set of linearly
independent quantum states can be unambiguously discriminated \cite{Chefles-1998}.
This simply means that one can always find a measurement, which may
not be optimal, that unambiguously discriminates all states. Second,
a pure bipartite entangled state with $d$ nonzero Schmidt coefficients
can be converted, with some nonzero probability, to a maximally entangled
state in $d\otimes d$ by LOCC. The optimal probability of such a
local conversion can be obtained using the result in \cite{Lo-Popescu-2001,Vidal-1999}
and is stated in the following lemma (proof in appendix). 

\textbf{Lemma 1.}\emph{ Let $\left|\Psi\right\rangle _{AB}=\sum_{i=1}^{d}\sqrt{\alpha_{i}}\left|i\right\rangle _{A}\left|i\right\rangle _{B}$
be a bipartite pure entangled state, where $\left\{ \sqrt{\alpha_{i}}\right\} $
are the Schmidt coefficients such that $\alpha_{1}\geq...\geq\alpha_{d}>0$
and $\sum_{i=1}^{d}\alpha_{d}=1$. Then the optimal probability with
which $\left|\Psi\right\rangle _{AB}$ can be locally converted to
a maximally entangled state in $d\otimes d$ is given by $d\alpha_{d}$.} 
\begin{proof}
(Theorem 1) For convenience we first sketch the main idea behind the
proof. We shall begin with a scenario of local conversion between
two bipartite states (say, source and target), where the target state
is maximally entangled. The source state is so constructed that (a)
any measurement, say, $\boldsymbol{\Pi}$, on Alice's side that unambiguously
discriminates the states $\left\{ \left|\psi_{j}\right\rangle \left|j=1,\dots,N\right.\right\} $
constitutes a local protocol for the aforementioned state transformation,
and (b) the probability of local conversion, say, $P\left(\boldsymbol{\Pi}\right)$,
thus obtained is exactly equal to the average probability of success
in an unambiguous discrimination scenario, where the measurement $\boldsymbol{\Pi}$
distinguishes the states $\left\{ p_{j},\left|\psi_{j}\right\rangle \left|j=1,\dots,N\right.\right\} $.
However, for any $\boldsymbol{\Pi}$, $P\left(\boldsymbol{\Pi}\right)$
is bounded by the optimal local conversion probability obtained from
Lemma 1. An upper bound on the optimal average probability of success
follows by choosing Alice's measurement to be optimal for unambiguous
discrimination, that is, $\boldsymbol{\Pi}=\boldsymbol{\Pi}^{\mbox{opt}}$.
Further refinement leads us to inequality (\ref{thm-1}). We now give
the formal proof in three key steps. \textbf{}\\
(i)\textbf{ }Consider a bipartite scenario with two spatially
separated observers, Alice and Bob, with Alice holding quantum systems
$A_{1}$ and $A_{2}$ of dimensions $N^{\prime}\geq N$ and $N$,
respectively, and Bob holding a quantum system $B$ of dimension $N$.
Alice and Bob share the following pure state \begin{eqnarray}
\left|\psi^{\boldsymbol{\theta}}\right\rangle _{AB} & = & \sum_{j=1}^{N}\sqrt{p_{j}}e^{i\theta_{j}}\left|\psi_{j}\right\rangle _{A_{1}}\otimes\left|\Phi_{j}\right\rangle _{A_{2}B},\label{psiAB}\end{eqnarray}
where $\boldsymbol{\theta}$ represents the collection of parameters
$\left\{ \theta_{i}\left|i=1,\dots,N\right.\right\} $ allowed to
vary and $\left\{ \left|\Phi_{j}\right\rangle \vert j=1,\dots,N\right\} $
is a set of $N$ mutually orthonormal maximally entangled states in
$N\otimes N$ defined as\begin{eqnarray}
\left|\Phi_{j}\right\rangle  & = & \frac{1}{\sqrt{N}}\sum_{k=1}^{N}\exp\left(\frac{2\pi i\left(k-1\right)\left(j-1\right)}{N}\right)\left|k\right\rangle \vert k\rangle,\;\, j=1,\dots,N.\label{MES}\end{eqnarray}

Suppose Alice and Bob wish to convert $\left|\psi^{\boldsymbol{\theta}}\right\rangle _{AB}$
to a maximally entangled state, say, $\left|\Phi\right\rangle _{AB}$
in $N\otimes N$, by LOCC. This can be achieved by a local protocol,
which is not necessarily optimal, where Alice performs a generalized
measurement (POVM) $\boldsymbol{\Pi}=\left\{ \Pi_{k}\left|k=1,\dots,N+1\right.\right\} $
on system $A_{1}$ that unambiguously discriminates the states $\left\{ \left|\psi_{j}\right\rangle \left|j=1,\dots,N\right.\right\} $.
The POVM elements $\left\{ \Pi_{k}\left|k=1,\dots,N\right.\right\} $
satisfy Eq.$\,$(\ref{unambiguous-M}), where the outcomes $j=1,\dots,N$
correspond to success and the outcome $j=N+1$ corresponds to failure.
If the outcome is $j$, the measurement successfully detects the state
$\left|\psi_{j}\right\rangle $ for $j=1,\dots,N$. From the expression
of $\left|\psi^{\boldsymbol{\theta}}\right\rangle _{AB}$ given by
Eq.$\,$(\ref{psiAB}) it is evident that this occurs with probability
$p_{j}\gamma_{j}$, and for each of these cases the corresponding
maximally entangled state $\left|\Phi_{j}\right\rangle $ is created
between Alice and Bob. For $j=N+1$, the outcome is inconclusive,
and therefore will not be our concern. 

The above local protocol, with some nonzero probability, converts
the state $\left|\psi^{\boldsymbol{\theta}}\right\rangle _{AB}$ to
a maximally entangled state in $N\otimes N$. Note that for every
successful outcome, the maximally entangled state created between
Alice and Bob, can be converted to the designated state $\left|\Phi\right\rangle _{AB}$
by local unitaries. Thus the probability of creating a maximally entangled
state between Alice and Bob with this local protocol is $P\left(\boldsymbol{\Pi}\right)=\sum_{j=1}^{N}p_{j}\gamma_{j}$,
which is the same as the average probability of success in unambiguous
discrimination of the states $\left\{ p_{j},\left|\psi_{j}\right\rangle \left|j=1,\dots,N\right.\right\} $
with the measurement $\boldsymbol{\Pi}$. 

Now suppose that the POVM $\boldsymbol{\Pi}=\boldsymbol{\Pi}_{\mbox{opt}}$,
that is, the measurement is optimal for unambiguous discrimination
of the states $\left\{ p_{j},\left|\psi_{j}\right\rangle \left|j=1,\dots,N\right.\right\} $.
Then $P_{\mbox{opt}}=\sum_{j=1}^{N}p_{j}\gamma_{j}^{\mbox{opt}}$,
which by our previous argument is also the probability, not necessarily
optimal, of locally converting the state $\left|\psi^{\boldsymbol{\theta}}\right\rangle _{AB}$
to $\left|\Phi\right\rangle _{AB}$. However, $P_{\mbox{opt }}$ cannot
exceed the optimal local conversion probability $p\left(\psi_{AB}^{\boldsymbol{\theta}}\rightarrow\Phi_{AB}\right)$
that can be obtained by applying Lemma 1. Therefore, \begin{eqnarray}
P_{\mbox{opt}} & \leq & p\left(\psi_{AB}^{\boldsymbol{\theta}}\rightarrow\Phi_{AB}\right).\label{Plessthanp}\end{eqnarray}
(ii) To obtain an expression for $p\left(\psi_{AB}^{\boldsymbol{\theta}}\rightarrow\Phi_{AB}\right)$
we first write $\left|\psi^{\boldsymbol{\theta}}\right\rangle _{AB}$
in its Schimdt-decomposed form:\begin{eqnarray}
\left|\psi^{\boldsymbol{\theta}}\right\rangle  & = & \frac{1}{\sqrt{N}}\sum_{k=1}^{N}\left\Vert \left|\eta_{k}\right\rangle \right\Vert \left|\eta_{k}^{\prime}k\right\rangle _{A_{1}A_{2}}\left|k\right\rangle _{B},\label{psiAB-2}\end{eqnarray}
where $\left|\eta_{k}\right\rangle $ (unnormalized) is given by \begin{eqnarray}
\left|\eta_{k}\right\rangle  & = & \sum_{r=1}^{N}\sqrt{p_{r}}e^{i\theta_{r}}\exp\left(\frac{2\pi i}{N}\left(r-1\right)\left(k-1\right)\right)\left|\psi_{r}\right\rangle \label{eta_k}\end{eqnarray}
and $\left|\eta_{k}^{\prime}\right\rangle =\frac{1}{\left\Vert \left|\eta_{k}\right\rangle \right\Vert }\left|\eta_{k}\right\rangle $
is the normalized state. Observe that (\ref{psiAB-2}) is indeed the
Schmidt decomposition of $\left|\psi^{\boldsymbol{\theta}}\right\rangle _{AB}$
owing to $\left\langle \eta_{k}^{\prime}k\vert\eta_{m}^{\prime}m\right\rangle =\left\langle \eta_{k}^{\prime}\vert\eta_{m}^{\prime}\right\rangle \left\langle k\vert m\right\rangle =\left\langle \eta_{k}^{\prime}\vert\eta_{m}^{\prime}\right\rangle \delta_{km}$.
The Schmidt coefficients are given by $\frac{\left\Vert \left|\eta_{k}\right\rangle \right\Vert }{\sqrt{N}},\; k=1,\dots,N$,
where for every $k$, $\left\Vert \left|\eta_{k}\right\rangle \right\Vert >0$.
Thus all Schmidt coefficients are nonzero. Then from Lemma 1 it follows
that \begin{eqnarray}
p\left(\psi_{AB}^{\boldsymbol{\theta}}\rightarrow\Phi_{AB}\right) & = & \min_{k}\left\{ \left\Vert \left|\eta_{k}\right\rangle \right\Vert ^{2}\left|k=1,\dots,N\right.\right\} .\label{applying-vidal}\end{eqnarray}
It can be easily seen that $p\left(\psi_{AB}^{\boldsymbol{\theta}}\rightarrow\Phi_{AB}\right)$
depends on $\left\{ \theta_{i}\left|i=1,\dots,N\right.\right\} $,
inner products of the states $\left\{ \left|\psi_{r}\right\rangle \left|r=1,\dots,N\right.\right\} $,
and the probabilities $\left\{ p_{i}\left|i=1,\dots,N\right.\right\} $.
Of all these only the real parameters $\theta_{i}$ can be varied,
everything else remaining fixed for a given set $\left\{ p_{i},\left|\psi_{i}\right\rangle \left|i=1,\dots,N\right.\right\} $. 

Noting that $P_{\mbox{opt }}$ does not depend on $\left\{ \theta_{i}\right\} $,
inequality (\ref{Plessthanp}) therefore holds for \emph{any} set
$\left\{ \theta_{i}\right\} $, and in particular any set that minimizes
$p\left(\psi_{AB}^{\boldsymbol{\theta}}\rightarrow\Phi_{AB}\right)$.
Therefore, \begin{eqnarray}
P_{\mbox{opt}} & \leq & \min_{\left\{ \theta_{i}\right\} }p\left(\psi_{AB}^{\boldsymbol{\theta}}\rightarrow\Phi_{AB}\right),\label{min- Plessthanp}\end{eqnarray}
gives us the best possible bound on $P_{\mbox{opt}}$ using this approach.
\\
(iii) To evaluate the right-hand.side of (\ref{min- Plessthanp})
we proceed as follows. First, we observe that \begin{eqnarray}
\min_{\left\{ \theta_{i}\right\} }p\left(\psi_{AB}^{\boldsymbol{\theta}}\rightarrow\Phi_{AB}\right) & = & \min_{k}\left\{ \min_{\left\{ \theta_{i}\right\} }\left\Vert \left|\eta_{k}\right\rangle \right\Vert ^{2}\left|k=1,\dots,N\right.\right\} .\label{applying-vidal-refined}\end{eqnarray}
Next, we prove the following equality: \begin{eqnarray}
\min_{\left\{ \theta_{i}\right\} }\left\Vert \left|\eta_{k}\right\rangle \right\Vert ^{2} & = & \min_{\left\{ \theta_{i}\right\} }\left\Vert \left|\eta_{j}\right\rangle \right\Vert ^{2},\label{etak-etaj}\end{eqnarray}
for every pair $\left(k,j\right)$. To prove Eq.$\,$(\ref{etak-etaj})
we first express $\left\Vert \left|\eta_{k}\right\rangle \right\Vert ^{2}$
as,\begin{eqnarray}
\left\Vert \left|\eta_{k}\right\rangle \right\Vert ^{2} & = & \left\Vert \sum_{r=1}^{N}\sqrt{p_{r}}e^{i\theta_{r}^{\prime}(k)}\left|\psi_{r}\right\rangle \right\Vert ^{2},\label{eq:etak-last}\end{eqnarray}
where $\theta_{r}^{\prime}\left(k\right)=\theta_{r}+\frac{2\pi}{N}\left(r-1\right)\left(k-1\right)$
for $r=1,\dots,N$. Now suppose that the set $\left\{ \theta_{r}\left|r=1,\dots,N\right.\right\} $
minimizes $\left\Vert \left|\eta_{k}\right\rangle \right\Vert ^{2}$.
Noting that (\ref{eq:etak-last}) has exactly the same form of $\left\Vert \left|\eta_{1}\right\rangle \right\Vert ^{2}$,
the set $\left\{ \theta_{r}^{\prime}\left(k\right)\left|r=1,\dots,N\right.\right\} $
therefore minimizes $\left\Vert \left|\eta_{1}\right\rangle \right\Vert ^{2}$.
A similar argument holds for every $i,i\neq k$. We have therefore
proved (\ref{etak-etaj}) and consequently, \begin{eqnarray}
\min_{\left\{ \theta_{i}\right\} }p\left(\psi_{AB}^{\boldsymbol{\theta}}\rightarrow\Phi_{AB}\right) & = & \min_{\left\{ \theta_{i}\right\} }\left\Vert \left|\eta_{k}\right\rangle \right\Vert ^{2}\;\mbox{\ensuremath{\forall}}k=1,\dots,N\label{main-inequality}\end{eqnarray}
Inequalities (\ref{min- Plessthanp}) and (\ref{main-inequality})
for $k=1$ together prove the theorem. 
\end{proof}
We now show that the upper bound in (\ref{thm-1}) is saturated when
the optimal solution is a nonsingular interior point of the critical
feasible region $\mathcal{R}$. Therefore in a generic case the optimal
average probability of success is equal to the upper bound given by
Theorem 1. 
\begin{thm}
Let a quantum system be prepared in one of the linearly independent
pure states $\left|\psi_{1}\right\rangle ,\dots,\left|\psi_{N}\right\rangle $
with prior probabilities $p_{1},\dots,p_{N}$ respectively, where
$0<p_{i}<1$ for every $i$ and $\sum_{i=1}^{N}p_{i}=1$. For an optimal
unambiguous state discrimination measurement, suppose that the solution
is an interior nonsingular point of the critical feasible region.
Then \begin{eqnarray}
P_{\mbox{opt}} & = & \min_{\left\{ \theta_{i}\right\} }\left\Vert \sum_{j=1}^{N}\sqrt{p_{j}}e^{i\theta_{j}}\left|\psi_{j}\right\rangle \right\Vert ^{2}.\label{eq:thm2}\end{eqnarray}
\end{thm}
\begin{proof}
In \cite{Pang-Wu-2009} it was shown that if $\boldsymbol{\gamma}^{\mbox{opt}}$
is an interior nonsingular point of the critical feasible region,
then $P_{\mbox{opt}}$ can be expressed as \begin{eqnarray}
P_{\mbox{opt}} & = & \left\Vert \sum_{j}\sqrt{p_{j}}e^{i\theta_{j}^{\prime}}\left|\psi_{j}\right\rangle \right\Vert ^{2},\label{eq:proof-thm-2-eq-1}\end{eqnarray}
but no explicit expressions of the phases $e^{i\theta_{j}^{\prime}}$
were given. However, it was noted that $P_{\mbox{opt }}$ must be
the value of a stationary point if the phases are allowed to change
freely. Note that without the explicit knowledge of the phases or
knowing how to obtain them (a stationary point may be a minimum or
maximum), Eq.$\,$(\ref{eq:proof-thm-2-eq-1}) is not very useful. 

However, the upper bound in Theorem 1 {[}inequality$\,$(\ref{thm-1}){]}
which holds irrespective of the class of optimal solution, fills this
gap. Consequently, \begin{eqnarray}
P_{\mbox{opt}} & = & \min_{\left\{ \theta_{i}\right\} }\left\Vert \sum_{j}\sqrt{p_{j}}e^{i\theta_{j}}\left|\psi_{j}\right\rangle \right\Vert ^{2}.\label{eq:proof-thm-2-eq-3}\end{eqnarray}
This proves the theorem and also shows that the stationary point must
be a minimum. 
\end{proof}
It is not clear whether our bound saturates for the other two classes
of optimal solution, as the expression (\ref{eq:proof-thm-2-eq-1})
was obtained \cite{Pang-Wu-2009} assuming that the optimal solution
is an interior nonsingular point of the critical feasible region.
However, by considering examples from each of the other two classes,
we first show that the upper bound given by Theorem 1 is tight for
both. The third example shows that for a boundary solution, the optimal
value could be strictly less than the value obtained from our bound.
Therefore, an optimal boundary solution will not in general saturate
our bound. 

\emph{Example I:} \emph{boundary point.} We begin by considering an
example for $N=3$ \cite{Sun-et-al-2001}, where the given states,
$\left|\psi_{1}\right\rangle =\left(\begin{array}{ccc}
1 & 0 & 0\end{array}\right)^{T}$,$\left|\psi_{2}\right\rangle =\sqrt{\frac{1}{3}}\left(\begin{array}{ccc}
1 & 1 & 1\end{array}\right)^{T}$, and $\left|\psi_{3}\right\rangle =\frac{1}{\sqrt{3}}\left(\begin{array}{ccc}
1 & 1 & -1\end{array}\right)^{T}$ are equally likely. Noting that the inner products are all real,
inequality (\ref{eq:P-opt-2}) becomes (set $\theta_{1}=0$)\begin{eqnarray*}
P_{\mbox{opt}} & \leq & 1+\min_{\left\{ \theta_{2},\theta_{3}\right\} }\frac{2}{3\sqrt{3}}\left[\cos\theta_{2}+\cos\theta_{3}+\frac{1}{\sqrt{3}}\cos\left(\theta_{3}-\theta_{2}\right)\right]\end{eqnarray*}
By simple numerical minimization using Mathematica we find that $P_{\mbox{opt}}\leq0.4444$.
In \cite{Sun-et-al-2001} it was shown that $\boldsymbol{\gamma}^{\mbox{opt}}=\left\{ 0,\frac{2}{3},\frac{2}{3}\right\} $,
from which we obtain $P_{\mbox{opt}}=\frac{4}{9}=0.4444$. Thus the
upper bound is achieved. As one of the individual success probabilities
is zero, the optimum point is therefore on the boundary. 

It is easy to construct an example for any $N\geq4$ starting from
the one we just discussed. Here we give an example for $N=4$, from
which it will be evident how to generalize for higher $N$. Consider
the set of states $\left\{ \left|\psi_{i}\right\rangle \left|i=1,\dots,4\right.\right\} $,
where the first three states are from the above example, and the new
state $\left|\psi_{4}\right\rangle $ has the property that $\left|\psi_{4}\right\rangle \perp\left|\psi_{i}\right\rangle $
for $i=1,2,3$. We choose the prior probabilities as $p_{i}=\frac{1-p}{3}$
for $i=1,2,3$ and $p_{4}=p$, where $0<p<1$. In this case, using
inequality (\ref{eq:P-opt-2}) we find that \begin{eqnarray*}
P_{\mbox{opt}} & \leq & p+0.4444\left(1-p\right)\end{eqnarray*}
 To show that the above bound is tight, we find the optimal set of
the individual success probabilities. Noting that $\left|\psi_{4}\right\rangle $
is orthogonal to every other state, it is easy to obtain that $\gamma^{\mbox{opt}}=\left\{ 0,\frac{2}{3},\frac{2}{3},1\right\} $
and $P_{\mbox{opt}}=p+\frac{4}{9}\left(1-p\right)$, thereby achieving
the upper bound. 

\emph{Example II:} \emph{interior singular point}. We begin by considering
such an example for $N=3$ \cite{Pang-Wu-2009}. Consider the following
vectors $\left|\psi_{1}\right\rangle =\left(\begin{array}{ccc}
1 & 0 & 0\end{array}\right)^{T}$, $\left|\psi_{2}\right\rangle =\sqrt{\frac{1}{5}}\left(\begin{array}{ccc}
1 & 2 & 0\end{array}\right)^{T}$ and $\left|\psi_{3}\right\rangle =\frac{2}{\sqrt{17}}\left(\begin{array}{ccc}
1 & 1 & \frac{3}{2}\end{array}\right)^{T},$ with prior probabilities $p_{1}=0.30$, $p_{2}=0.35$, and $p_{3}=0.35$
respectively. We see that the inner products are real. Using inequality
(\ref{eq:P-opt-2}), a simple numerical minimization using Mathematica
shows that $P_{\mbox{opt}}\leq0.4430$ which agrees with the optimal
value \cite{Pang-Wu-2009}. Following the method used in the previous
example, we can therefore generalize this example for any $N\geq4$. 

\emph{Example III:}\textbf{ }In this example we show that the upper
bound does not saturate in general for an optimal boundary solution.
This example is from \cite{Pang-Wu-2009}, where the states $\left|\psi_{1}\right\rangle =\left(\begin{array}{ccc}
1 & 0 & 0\end{array}\right)^{T}$, $\left|\psi_{2}\right\rangle =\sqrt{\frac{1}{5}}\left(\begin{array}{ccc}
1 & 2 & 0\end{array}\right)^{T}$ and $\left|\psi_{3}\right\rangle =\frac{2}{\sqrt{17}}\left(\begin{array}{ccc}
1 & 1 & \frac{3}{2}\end{array}\right)^{T}$ occur with prior probabilities $p_{1}=0.10$, $p_{2}=0.80$, and
$p_{3}=0.10$ respectively. Once again using inequality (\ref{eq:P-opt-2}),
a simple numerical minimization using Mathematica shows that $P_{\mbox{opt}}\leq0.4758$,
which is pretty close to the optimal value $P_{\mbox{opt}}=0.4632$
\cite{Pang-Wu-2009}. 

To conclude, we studied the problem of unambiguous discrimination
of $N$ linearly independent pure quantum states, where the measurement
strategy is such that either the input state is correctly identified
(zero error) or we learn nothing about it. The objective is to find
a measurement that maximizes the average probability of success. This
problem has been extensively studied over the years, but the exact
solution is known only for $N=2$, and special cases for $N\geq3$.
In this paper we obtained an upper bound on the optimal average probability
of success using a result \cite{Lo-Popescu-2001,Vidal-1999} on optimal
local conversion between two bipartite pure states. We showed that
for $N\geq2$ an optimal measurement in general saturates our bound,
thereby providing an exact expression of the optimal average probability
of success in the generic case. In the exceptional cases we have shown
that the bound is tight, but not always attained for an optimal boundary
solution. \\

\begin{acknowledgement*}
The author is grateful to Sibasish Ghosh and Michael Nathanson for
comments on the manuscript and very useful discussions. Discussions
with Rahul Jain, R Rajesh and Manik Banik are gratefully acknowledged.
This research is supported in part by DST-SERB project SR/S2/LOP-18/2012. \end{acknowledgement*}

\section*{APPENDIX}

\section{Proof of Lemma 1}

Let\emph{ $\left|\Psi\right\rangle _{AB}=\sum_{i=1}^{d}\sqrt{\alpha_{i}}\left|i\right\rangle _{A}\left|i\right\rangle _{B}$
}be a bipartite pure entangled state, where $\left\{ \sqrt{\alpha_{i}}\right\} $
are the Schmidt coefficients such that $\alpha_{1}\geq\dots\geq\alpha_{d}>0$
and $\sum_{i=1}^{d}\alpha_{d}=1.$ Let $|\Phi\rangle$ be a maximally
entangled state in $d\otimes d$. From Vidal's Theorem \cite{Vidal-1999},
the optimal probability of local conversion $|\Psi\rangle\rightarrow|\Phi\rangle$
is given by\[
P(\Psi_{AB}\rightarrow\Phi)=\min_{l\in[1,d]}q_{l},\]
where \[
q_{l}=\frac{\sum_{i=l}^{d}\alpha_{i}}{\frac{1}{d}(d-l+1)}.\]

Because $\alpha_{1}\geq\alpha_{2}\geq...\geq\alpha_{d}>0$, we have
$\sum_{i=l}^{d}\alpha_{i}\geq(d-l+1)\alpha_{d}$. Therefore, for every
$l$, $l=1,...,d-1$, we have \begin{eqnarray*}
q_{l} & = & \frac{\sum_{i=l}^{d}\alpha_{i}}{\frac{1}{d}(d-l+1)}\\
 & \geq & \frac{(d-l+1)\alpha_{d}}{\frac{1}{d}(d-l+1)}\\
 & = & d\alpha_{d}\\
 & = & q_{d}.\end{eqnarray*}
 This completes the proof.

\section{Classes of Optimal Solution }

Here we define all possible classes of optimal solution following
\cite{Pang-Wu-2009}. Consider the sets of individual success probabilities
$\left\{ \gamma_{1},\gamma_{2},\dots,\gamma_{N}\right\} $ and the
prior probabilities $\left\{ p_{1},p_{2},\dots,p_{N}\right\} $ as
vectors $\boldsymbol{\gamma}$ and $\mathbf{p}$ respectively in the
$N$ dimensional real vector space $\mathbb{R}^{N}$.  For convenience
(and to avoid any confusion) we adapt the nomenclature of \cite{Pang-Wu-2009}
and refer to these vectors as {}``points''. Now for a given set
of states, the set of possible optimal solutions is determined by
the constraints imposed by the problem of unambiguous discrimination. 

Define the matrices $X=\Lambda^{\dagger}\Lambda$ and $\Gamma$, where
$\Lambda$ is the matrix whose $i^{\mbox{th }}$ column is $\left|\psi_{i}\right\rangle $
and $\Gamma\geq0$ is a diagonal matrix, whose diagonal elements are
the success probabilities $\gamma_{i}$. It was shown \cite{Pang-Wu-2009}
that the set of points \textbf{$\boldsymbol{\gamma}$} (denote by
$S$) satisfying the constraints $X-\Gamma\geq0$ and $\Gamma\geq0$
imposed by the problem, is convex. The set $S$ is said to be the
feasible set. The critical feasible region $\mathcal{R}$ is defined
as the set of points $\boldsymbol{\gamma}\in S$ satisfying $\sigma_{\min}\left(\boldsymbol{\gamma}\right)=0$,
where $\sigma_{\min}$ is the minimum eigenvalue of $X-\Gamma$. This
set is closed. Note that the critical feasible region is in fact the
set of candidate optimal solutions and is fixed for a given set of
states. Once we specify the prior probabilities, the optimal solution
becomes unique in the sense that there is no other solution which
is also optimal for the same set of prior probabilities. Different
sets of prior probabilities in general lead to different optimal solutions
within the set $\mathcal{R}$. 

It was shown \cite{Pang-Wu-2009} that the optimal solution $\boldsymbol{\gamma}^{\mbox{opt}}\in\mathcal{R}$
is either an interior nonsingular point (that is, $\nabla\sigma_{n}\left(\boldsymbol{\gamma}\right)\vert_{\boldsymbol{\gamma}^{\mbox{opt}}}=-\mathbf{p}$),
or an interior singular point($\nabla\sigma_{n}\left(\boldsymbol{\gamma}\right)=\mathbf{0}$),
or a point on the boundary of $\mathcal{R}$. If it is an interior
point (nonsingular or singular) then it means that the optimal measurement
is able to discriminate all states, that is, for every $i$, $0<\gamma_{i}^{\mbox{opt}}\leq1$.
On the other hand, if it is a boundary point, then at least one of
the optimal individual success probabilities is zero. Moreover, an
interior nonsingular optimal solution is nondegenerate, i.e., it's
the optimal solution for an unique set of prior probabilities, wheras
an interior singular point solution is degenerate, which implies that
it can be the optimal solution for different sets of prior probabilities.
It should be noted that interior singular points are exceptions and
may not even exist for a given set of states. For the necessary and
sufficient conditions pertaining to these optimal solutions and further
details please see \cite{Pang-Wu-2009}.

\section{Examples}

Here we illustrate with several examples where our bound is saturated.

\subsection{Two states }

For two states inequality (\ref{eq:P-opt-2}) reduces to \begin{eqnarray*}
P_{\mbox{opt}} & \leq & 1+\min_{\theta_{1},\theta_{2}}2\sqrt{p_{1}p_{2}}\left|\left\langle \psi_{1}\left|\psi_{2}\right.\right\rangle \right|\cos\left(\theta_{2}-\theta_{1}+\phi_{12}\right).\end{eqnarray*}
Because $\phi_{12}$ is fixed, the minimum is clearly given by choosing
$\theta_{2},\theta_{1}$ such that $\theta_{2}-\theta_{1}+\phi_{12}=\pi$.
Set $\theta_{1}=0$, and $\theta_{2}=\pi-\phi_{12}$ yielding\begin{eqnarray}
P_{\mbox{opt}} & \leq & 1-2\sqrt{p_{1}p_{2}}\left|\left\langle \psi_{1}\left|\psi_{2}\right.\right\rangle \right|.\label{two-state}\end{eqnarray}
The upper bound given by (\ref{two-state}) matches the IDP result
\cite{Ivanovic-1987,Dieks-1988,Peres-1988} obtained when the states
are equally likely and the more general result by Jaeger and Shimony
\cite{Jaeger-Shimony-1995} for unequal prior probabilities.

\subsection{Three states}

The case $N=3$ has been extensively studied, but an analytical solution
is not known except for special cases. Here we consider several examples
from the literature, and show that our bound is tight in each case.
We first write (\ref{eq:P-opt-2}) explicitly for $N=3$, where without
loss of generality, we have set $\theta_{1}=0$. \begin{eqnarray}
P_{\mbox{opt}} & \leq & 1+\min_{\theta_{2},\theta_{3}}2\left[\sqrt{p_{1}p_{2}}\left|\left\langle \psi_{1}\left|\psi_{2}\right.\right\rangle \right|\cos\left(\theta_{2}+\phi_{12}\right)+\sqrt{p_{1}p_{3}}\left|\left\langle \psi_{1}\left|\psi_{3}\right.\right\rangle \right|\cos\left(\theta_{3}+\phi_{13}\right)\right.\nonumber \\
 &  & +\left.\sqrt{p_{2}p_{3}}\left|\left\langle \psi_{2}\left|\psi_{3}\right.\right\rangle \right|\cos\left(\theta_{3}-\theta_{2}+\phi_{23}\right)\right].\label{P-opt-three-states}\end{eqnarray}

\subsubsection*{Example 1 }

Suppose that $\left\langle \psi_{1}\left|\psi_{2}\right.\right\rangle =0$,
but $\left\langle \psi_{1}\left|\psi_{3}\right.\right\rangle \neq0$,
$\left\langle \psi_{2}\left|\psi_{3}\right.\right\rangle \neq0$,
then inequality (\ref{P-opt-three-states}) becomes \begin{eqnarray}
P_{\mbox{opt}} & \leq & 1+\min_{\theta_{2},\theta_{3}}2\left[\sqrt{p_{1}p_{3}}\left|\left\langle \psi_{1}\left|\psi_{3}\right.\right\rangle \right|\cos\left(\theta_{3}+\phi_{13}\right)+\sqrt{p_{2}p_{3}}\left|\left\langle \psi_{2}\left|\psi_{3}\right.\right\rangle \right|\cos\left(\theta_{3}-\theta_{2}+\phi_{23}\right)\right].\label{ex-1-eq-2}\end{eqnarray}
The minimum of the right hand side is given by $\theta_{2},\theta_{3}$
satisfying\begin{eqnarray*}
\cos\left(\theta_{3}+\phi_{13}\right) & = & -1\\
\cos\left(\theta_{3}-\theta_{2}+\phi_{23}\right) & = & -1\end{eqnarray*}
The above two equations are satisfied for $\theta_{3}=\pi-\phi_{13}$,
and $\theta_{2}=\phi_{23}-\phi_{13}$. We therefore have \begin{eqnarray}
P_{\mbox{opt}} & \leq & 1-2\left[\sqrt{p_{1}p_{3}}\left|\left\langle \psi_{1}\left|\psi_{3}\right.\right\rangle \right|+\sqrt{p_{2}p_{3}}\left|\left\langle \psi_{2}\left|\psi_{3}\right.\right\rangle \right|\right].\label{ex-1-eq-1}\end{eqnarray}
which matches the optimal value obtained in \cite{Pang-Wu-2009}.

\subsubsection*{Example 2}

This example is from \cite{Bergou-Fut-Feld-2012}, where the authors
introduced the invariant phase also known as the geometric phase.
Denote the complex overlaps of the states as $\left\langle \psi_{1}\left|\psi_{2}\right.\right\rangle =\left|\left\langle \psi_{1}\left|\psi_{2}\right.\right\rangle \right|e^{i\phi_{3}}$
and two more cyclic permutation of the indices. The invariant phase
$\phi$, defined as $\phi=\phi_{1}+\phi_{2}+\phi_{3}$ corresponds
to the phase deficit associated with a closed path in the parameter
space. 

To make the connection explicit we rewrite inequality (\ref{P-opt-three-states})
with the following substitutions $\phi_{12}=\phi_{3}$, $-\phi_{13}=\phi_{2}$
and $\phi_{23}=\phi_{1}$: \begin{eqnarray}
P_{\mbox{opt}} & \leq & 1+\min_{\theta_{2},\theta_{3}}2\left[\sqrt{p_{1}p_{2}}\left|\left\langle \psi_{1}\left|\psi_{2}\right.\right\rangle \right|\cos\left(\theta_{2}+\phi_{3}\right)+\sqrt{p_{1}p_{3}}\left|\left\langle \psi_{1}\left|\psi_{3}\right.\right\rangle \right|\cos\left(\theta_{3}-\phi_{2}\right)\right.\nonumber \\
 &  & +\left.\sqrt{p_{2}p_{3}}\left|\left\langle \psi_{2}\left|\psi_{3}\right.\right\rangle \right|\cos\left(\theta_{3}-\theta_{2}+\phi_{1}\right)\right].\label{P-opt-three-states-01}\end{eqnarray}
 Noting that $\phi=\sum_{i=1}^{3}\phi_{i}$, \begin{eqnarray}
P_{\mbox{opt}} & \leq & 1+\min_{\theta_{2},\theta_{3}}2\left[\sqrt{p_{1}p_{2}}\left|\left\langle \psi_{1}\left|\psi_{2}\right.\right\rangle \right|\cos\alpha+\sqrt{p_{1}p_{3}}\left|\left\langle \psi_{1}\left|\psi_{3}\right.\right\rangle \right|\cos\beta\right.\nonumber \\
 &  & +\left.\sqrt{p_{2}p_{3}}\left|\left\langle \psi_{2}\left|\psi_{3}\right.\right\rangle \right|\cos\left(\alpha-\beta+\phi\right)\right],\label{P-opt-three-states-02}\end{eqnarray}
where $\alpha=\theta_{2}+\phi_{3}$, $\beta=\theta_{3}-\phi_{2}$.
When $\phi=\pi$, minimum of the r.h.s. of (\ref{P-opt-three-states-02})
is obtained for $\alpha=\beta=\pi$, giving the following upper bound:
\begin{eqnarray*}
P_{\mbox{opt}} & \leq & 1-2\left[\sqrt{p_{1}p_{2}}\left|\left\langle \psi_{1}\left|\psi_{2}\right.\right\rangle \right|+\sqrt{p_{1}p_{3}}\left|\left\langle \psi_{1}\left|\psi_{3}\right.\right\rangle \right|+\sqrt{p_{2}p_{3}}\left|\left\langle \psi_{2}\left|\psi_{3}\right.\right\rangle \right|\right],\end{eqnarray*}
which agrees with the optimal value \cite{Bergou-Fut-Feld-2012} obtained
when $\gamma_{i}>0$ for every $i$.

\subsubsection*{Example 3}

The \emph{a priori} probabilities are equal and the overlap of the
three states are real and equal \[
\left\langle \psi_{1}\left|\psi_{3}\right.\right\rangle =\left\langle \psi_{1}\left|\psi_{3}\right.\right\rangle =\left\langle \psi_{2}\left|\psi_{3}\right.\right\rangle =s\;:\;0<s<1\]
Inequality (\ref{P-opt-three-states}) becomes \begin{eqnarray*}
P_{\mbox{opt}} & \leq & 1+\frac{2s}{3}\times\min_{\theta_{2},\theta_{3}}\left[\cos\theta_{2}+\cos\theta_{3}+\cos\left(\theta_{3}-\theta_{2}\right)\right].\end{eqnarray*}
A simple minimization using Mathematica shows that $\min\left[\cos\theta_{2}+\cos\theta_{3}+\cos\left(\theta_{3}-\theta_{2}\right)\right]=-\frac{3}{2}$.
Therefore, \begin{eqnarray*}
P_{\mbox{opt}} & \leq & 1-s,\end{eqnarray*}
which matches the optimal value obtained in \cite{Sun-et-al-2001}.

\subsection{Four states }

Here we will consider unambiguous discrimination of four geometrically
uniform states with equal prior probabilities \cite{Eldar-2003}.
Geometrically uniform states are defined over a group $\mathcal{G}$
of unitary matrices and are obtained by a single generating vector.
Consider the group $\mathcal{G}$ of $N=4$ unitary matrices $U_{i}$
defined as:\[
\begin{array}{cccc}
U_{1}=I_{4}, & U_{2}=\left[\begin{array}{cccc}
1 & 0 & 0 & 0\\
0 & -1 & 0 & 0\\
0 & 0 & 1 & 0\\
0 & 0 & 0 & -1\end{array}\right], & U_{3}=\left[\begin{array}{cccc}
1 & 0 & 0 & 0\\
0 & 1 & 0 & 0\\
0 & 0 & -1 & 0\\
0 & 0 & 0 & -1\end{array}\right], & U_{4}=U_{2}U_{3}\end{array}\]
The states that we wish to discriminate are given by: $\left|\psi_{i}\right\rangle =U_{i}\left|\psi\right\rangle ;\, i=1,\dots,4$,
where $\left|\psi\right\rangle =\frac{1}{3\sqrt{2}}\left(\begin{array}{cccc}
2 & 2 & 1 & 3\end{array}\right)^{T}$. The states are assumed to be equally likely. Then from inequality
(\ref{eq:P-opt-2}) \begin{eqnarray*}
P_{\mbox{opt}} & \leq & 1+\min_{\left\{ \theta_{2},\theta_{3},\theta_{4}\right\} }\frac{1}{18}\left[-4\cos\theta_{2}-\cos\theta_{3}+\cos\theta_{4}+4\cos\left(\theta_{3}-\theta_{2}\right)-\cos\left(\theta_{4}-\theta_{2}\right)-4\cos\left(\theta_{4}-\theta_{3}\right)\right].\end{eqnarray*}
The r.h.s is numerically minimized using Mathematica and we find that
\begin{eqnarray*}
P_{\mbox{opt}} & \leq & 0.2222,\end{eqnarray*}
which is in agreement with the optimal value $\frac{2}{9}$ \cite{Eldar-2003}. 
\end{document}